\documentclass[a4paper,UKenglish,cleveref, autoref, thm-restate,authorcolumns]{lipics-v2019}

\nolinenumbers

\title{On LTL Model Checking for Low-Dimensional Discrete Linear Dynamical Systems}

\author{Toghrul Karimov}{Max Planck Institute for Software Systems,
	Saarland Informatics Campus, Germany}{toghs@mpi-sws.org}
{}{}

\author{Jo{\"e}l Ouaknine}{Max Planck Institute for Software Systems, Saarland Informatics Campus, Germany
	\and
	Department of Computer Science, University of Oxford, United Kingdom}
{joel@mpi-sws.org}{}{Supported by ERC grant AVS-ISS (648701) and DFG grant 389792660 as part of TRR
	248 (see https://perspicuous-computing.science).}

\author{James Worrell}{Department of Computer Science, University of
	Oxford, United Kingdom}{jbw@cs.ox.ac.uk}
{}{Supported by EPSRC Fellowship EP/N008197/1.}

\authorrunning{T. Karimov, J. Ouaknine and J.Worrell} 

\Copyright{Toghrul Karimov, Jo{\"e}l Ouaknine, James Worrell} 



\category{} 

\relatedversion{} 

\supplement{}

\EventEditors{Javier Esparza and Daniel Kr{\'a}l'}
\EventNoEds{2}
\EventLongTitle{45th International Symposium on Mathematical Foundations of Computer Science (MFCS 2020)}
\EventShortTitle{MFCS 2020}
\EventAcronym{MFCS}
\EventYear{2020}
\EventDate{August 24--28, 2020}
\EventLocation{Prague, Czech Republic}
\EventLogo{}
\SeriesVolume{170}
\ArticleNo{50}

\begin{document}

\maketitle

\begin{abstract}
	Consider a discrete dynamical system given by a square matrix $M \in
	\rationals^{d \times d}$ and a starting point $s \in
	\rationals^d$. The \emph{orbit} of such a system is the infinite
	trajectory $\langle s, Ms, M^2s, \ldots\rangle$. Given a collection
	$T_1, T_2, \ldots, T_m \subseteq \reals^d$ of semialgebraic sets, we
	can associate with each $T_i$ an atomic proposition $P_i$ which
	evaluates to \emph{true} at time $n$ if, and only if, $M^ns \in
	T_i$. This gives rise to the \emph{LTL Model-Checking Problem} for
	discrete linear dynamical systems: given such a system $(M,s)$ and an
	LTL formula over such atomic propositions, determine whether the orbit
	satisfies the formula.  The main contribution of the present paper is
	to show that the LTL Model-Checking Problem for discrete linear
	dynamical systems is decidable in dimension 3 or less.
\end{abstract}

\section{Introduction}
A \textit{discrete-time linear dynamical system} consists of a square
matrix $M \in\rationals^{d \times d}$ and a starting point $s \in
\rationals^d$. Its trajectory, the \emph{orbit} of $s$ under $M$,
is the infinite sequence $\langle s, Ms, M^2s, \ldots \rangle$.  Such
systems constitute a family of fundamental models, and decision
problems associated with their trajectories arise frequently in the
analysis of automata, Markov chains, linear recurrence sequences, and
linear while loops (see, e.g., \cite{ChonevOW16,TUCS05,
	kannan1986polynomial} and references therein).

One of the earliest decision problems for linear dynamical systems was
formulated by Harrison in 1969~\cite{harrison1969lectures}, and
subsequently baptised the \emph{``Orbit Problem''} by Kannan and Lipton,
who famously solved it a decade later~\cite{kannan1980orbit}. The
Orbit Problem asks, given a linear dynamical system $(M,s)$ in ambient
space $\reals^d$ together with a point target $t \in \rationals^d$,
whether the orbit of $s$ under $M$ reaches $t$. Kannan and Lipton
established polynomial-time decidability of the Orbit Problem in all
dimensions. In subsequent work~\cite{kannan1986polynomial}, Kannan and
Lipton speculated that more complex decision problems might also be
decidable; specifically, they considered variants of the Orbit Problem
in which the target $t$ is replaced by a linear subspace $T
\subseteq \reals^d$. They conjectured that for one-dimensional
subspaces, reachability might remain decidable, but in the same breath
they noted that when $T$ is a $(d-1)$-dimensional subspace of
$\reals^d$, the corresponding reachability problem is precisely
equivalent to the well-known \emph{Skolem Problem} asking whether a
linear recurrence sequence has a zero, which itself has been open for
many decades~\cite{TUCS05,tao2008structure} (although decidability is
known in dimension $4$ or less~\cite{MST84,Ver85}).

The problem of reaching a linear subspace was studied by
Chonev \emph{et al.}\ in~\cite{chonev2013orbit,ChonevOW16}, in which
the authors established decidability for subspaces of dimension up to
three (regardless of the dimension of the ambient space). Chonev \emph{et
	al.}\ then turned their attention to the \emph{Polyhedron-Hitting
	Problem}~\cite{chonev2015polyhedron}, in which the target is an
arbitrary polyhedron. Decidability in dimension $3$ was established,
but the authors showed that in dimensions $4$ or higher, solving the
Polyhedron-Hitting Problem would necessarily entail major
breakthroughs in Diophantine approximation, considered out of reach at
the present time. More recently, Almagor \emph{et al.}\ studied the
\emph{Semialgebraic Orbit Problem}, in which the target is an
arbitrary semialgebraic set~\cite{AOW19b}. Once again, decidability in
dimension $3$ was shown to hold.  Finally, in very recent work, and
building on earlier results~\cite{AOW17}, Almagor \emph{et
	al.}\ introduced a unifying framework for formulating
\emph{reachability} queries for discrete linear dynamical
systems~\cite{ShaullJournal}, subsuming all of the above
problems. Roughly speaking, the authors considered instances in which
both the source and target are semialgebraic sets, and a specification
formalism in which one may quantify over members of these
sets. Crucially, however, their \emph{First-Order Orbit Queries}
framework only allows \emph{reachability} queries (``is there a
positive integer $n$ such that, after $n$ steps, such and such
properties hold?''). Almagor \emph{et al.}\ established decidability
in dimension $3$.

\textbf{Main contributions.}  In this paper, going considerably beyond
reachability, we tackle the problem of full \emph{LTL model checking}
for orbits of discrete-time linear dynamical systems in dimension $3$
(or less). More precisely, we are given a linear dynamical system
$(M,s)$ (with a singleton starting point $s \in \rationals^3$),
together with a collection $T_1, \ldots, T_m \subseteq \reals^3$ of
semialgebraic sets, and an LTL formula $\varphi$ over atomic
propositions $P_1, \ldots, P_m$.  The atomic proposition
$P_i$ evaluates to \emph{true} at time $n$ if, and only if, the $n$-th
component of the orbit lies in $T_i$, i.e., $M^ns \in T_i$. Such a
framework enables one to formulate vastly more sophisticated and
complex properties of orbits than mere reachability. \textbf{Our main
	result is that the LTL Model-Checking Problem for discrete-time
	linear dynamical systems in three dimensions is decidable, with
	complexity in exponential space.}

Some remarks are in order:
\begin{enumerate}
	\item Since we have a single starting point, the orbit
	consists of a \emph{single} trajectory. The problem we are
	solving is sometimes referred to in the literature as
	\emph{``path checking''}, although typical applications in
	runtime verification and online monitoring involve finite
	traces, e.g., \cite{LS09}.  Path checking ultimately
	periodic infinite traces is considered in~\cite{MarkeyS03}, 
	but the traces arising from linear dynamical systems need
	not be ultimately periodic (see \cite{AAGT15}).
	
	\item Our framework is limited to dynamical systems in three (or
	fewer) dimensions. As mentioned earlier, it is known that mere
	polyhedral reachability is `hard' (in a Diophantine-approximation
	sense) in dimensions $4$ and above, and as LTL model checking with
	semialgebraic targets vastly generalises polyhedral reachability,
	one cannot reasonably expect to prove decidability in dimensions
	higher than $3$.
	
	\item Beyond the search for maximal versatility and generality, our
	use of semialgebraic sets---rather than, say, products of intervals
	or polyhedra---in our specification framework has a practical
	motivation: in application areas such as program analysis,
	semialgebraic sets are indispensable to formulate sufficiently
	expressive properties, whether one seeks to overapproximate a set of
	reachable states, or to synthesise invariants or barrier
	certificates; see, e.g., \cite{KCBR18}.
\end{enumerate}

On a technical level, our approach makes extensive use of spectral
techniques and relies on various tools from algebraic and
transcendental number theory, notably Baker's theorem on linear forms
in logarithms of algebraic numbers, as well as Kronecker's theorem in
Diophantine approximation.

In~\cite{AAGT15}, Agrawal \emph{et al.}\ consider a problem that is
closely related to ours, namely the approximate verification of the
symbolic dynamics of Markov chains. More specifically, they view a
Markov chain as a distribution transformer: a stochastic matrix $M$
and an initial probability distribution $s$ give rise to an orbit
$\langle s, Ms, M^2s, \ldots \rangle$. They further discretise the
probability space into finitely many boxes (products of intervals),
which give rise to atomic propositions in exactly the same manner as
in our setting. They then consider LTL model checking over
the resulting formalism, but observe that the set of infinite words
arising as symbolic trajectories of a given Markov chain can fail to
be $\omega$-regular; consequently, they switch their attention to
``$\epsilon$-approximations'' of the model-checking problem (the
precise definitions are technical) and are able to establish
decidability in \emph{all} dimensions.  This variant of the
model-checking problem does not allow to check a specific path and
thereby circumvents many of the difficulties arising in the present
paper.  The two pieces of work are therefore fairly distinct, both in terms of
their respective scope and in the mathematical approach taken, despite
sharing similar motivations.

\section{Mathematical background}
A semialgebraic set $T \subseteq \reals^n$ is defined by a Boolean
combination of polynomial inequalities of the form $p(x_1, \ldots,
x_n) \geq 0$ and $q(x_1, \ldots, x_n) > 0$ for polynomials $p, q \in
\mathbb{Z}[x_1, \ldots, x_n]$.  

\subsection{Algebraic numbers}
A complex number $\alpha$ is algebraic if it is a root of a polynomial $p$ with integer coefficients. We denote the set of algebraic numbers by $\algebraic$.
For an algebraic number $\alpha$, its defining polynomial $p_\alpha$ is the unique polynomial of the least degree that has $\alpha$ as a root and coefficients that do not share common factors. Given a polynomial $p \in \mathbb{Z}[x]$, let $\size{p}$ denote the bit length of its representation as a list of coefficients encoded in binary, $\deg(p)$ denote its degree and $H(p)$ denote its height (i.e. the maximum of the absolute value of coefficients of $p$). Throughout this work we make an extensive use of the facts that for each pair $\alpha$, $\beta$ of algebraic numbers, $\deg(\alpha \beta) \leq \deg(\alpha) + \deg(\beta)$ and $H(\alpha \beta) \leq H(\alpha) H(\beta)$. 

An algebraic number $\alpha$ can be represented using its defining polynomial $p_\alpha$ together with rational approximations of its real and imaginary parts to sufficient precision. More precisely, $\alpha$ can be represented by $(p_\alpha, a, b, r) \in \mathbb{Z}[x] \times \rationals^3$ provided that $\alpha$ is the unique root of $p_\alpha$ in the circle of radius $r$ around $a+bi$. A separation bound due to Mignotte \cite{mignotte1983some} asserts that for roots $\alpha \neq \beta$ of a polynomial $ p \in \mathbb{Z}[x]$,
\[
\length{\alpha-\beta} > \frac{\sqrt{6}}{d^{(d+1)/2}H^{d-1}}
\]
where $d, H$ are the degree and height of $p_\alpha$, respectively. Thus if $r$ is less than a quarter of the root separation bound, then the representation is well-defined. Given a polynomial $p \in \mathbb{Z}[x]$, we can compute a standard representation of each of its roots in time polynomial in $\size{p}$ \cite{basu2005algorithms}. Thus for an algebraic number $\alpha$, we denote by $\size{\alpha}$ the bit length of its standard representation.

Given representations of algebraic numbers $\alpha, \beta$ we can effectively compute representations for the algebraic numbers $\alpha + \beta$, $\alpha \beta$, $\frac{1}{\alpha}$, $\length{\alpha}$, $\operatorname{Re}(\alpha)$, $\operatorname{Im}(\alpha)$ in time polynomial in $\size{\alpha} + \size{\beta}$. Efficient algorithms for these tasks can be found in \cite{basu2005algorithms, cohen2013course}.

\subsection{Number-theoretic bounds}
Throughout the paper, we make an extensive use of the following lemma, which itself is a consequence of the celebrated Baker-Wüstholtz theorem.
\begin{lemma}[\cite{ouaknine2014ultimate}]
	There exists a constant $C$ such that for algebraic numbers $\alpha$, $\beta$, for every $n \geq 2$ if $\alpha^n \neq \beta$, then $\length{\alpha^n - \beta} \geq \frac{1}{n^{(\size{\alpha}+\size{\beta})^C}}.$
\end{lemma}
Lemma 2 below states the following: if we start at an arbitrary point $\gamma \in \torus$ on the unit circle, and repeatedly apply rotation through $\arg(\lambda)$ radians for $\lambda\in \torus \cap \algebraic$ that is not a root of unity, we will enter any open interval $J \subseteq \torus$ in at most a certain number of steps that does not depend on the starting point $\gamma$ but depends on the size of $J$. Henceforth we denote by $\length{J}$ the arc length of the interval $J \subseteq \torus$ in radians.
\begin{lemma}
	There exists a constant $D$ such that for every $\lambda \in \torus \cap \algebraic$ that is not a root of unity and open subinterval $J$ of $\torus$, for each $\gamma \in \torus$, $\gamma \lambda^n \in J$ for some $n < 2 \pi \left( \frac{2\pi}{\length{J}}
	\right)^{\size{\lambda}^D}$.
\end{lemma}
\begin{proof}
	By the Pigeonhole principle, if there are $N_1 > \frac{2\pi}{\length{J}}$ points on the unit circle, at least two of them will have arc distance smaller than $\length{J}$. We select $N_1 = 
	\left\lfloor \frac{2\pi}{\length{J}}
	\right\rfloor + 1 > \frac{2\pi}{\length{J}}$
	and consider the sequence $\langle \lambda, \lambda^2, \ldots \rangle$. By the preceding argument, there exist 
	$1 \leq  k < m  \leq N_1$ such that $\lambda^k$ and $\lambda^m$ have arc distance smaller than $\length{J}$. That is, $\arg(\lambda^{m-k}) < \length{J}$.
	
	Bounding arc length from below with Euclidean distance and using Lemma 1,
	\[
	\arg(\lambda^{m-k}) \geq
	\length{\lambda^m - \lambda^k} =
	\length{\lambda^{m-k}-1} 
	\geq 
	\frac{1}{(m-k)^{(\size{\lambda} + \size{1})^C}}
	\geq
	\frac{1}{\left\lfloor \frac{2\pi}{\length{J}}
		\right\rfloor^{(\size{\lambda} + \size{1})^C}}.
	\]
	Now consider the sequence $z_i = \gamma \lambda^{i(m-k)}$ for $i \geq 0$. As the arc distance between any consecutive terms $z_i, z_{i+1}$ is less than $\length{J}$, this sequence must enter $J$ before winding around the unit circle once. We therefore obtain that for some 
	\[
	n_1 \leq 
	\left\lfloor 
	\frac{2\pi - \length{J}}{\arg(\lambda^{m-k})}
	\right\rfloor + 1
	\leq
	\left\lfloor 
	(2\pi - \length{J}) 
	\left\lfloor \frac{2\pi}{\length{J}}
	\right\rfloor^{(\size{\lambda} + \size{1})^C}
	\right\rfloor + 1
	\leq
	2\pi 
	\left(
	\frac{2\pi}{\length{J}}
	\right)^{(\size{\lambda} + \size{1})^C}
	\]
	$z_{n_1} \in J$. Translating this back to the sequence $\langle \lambda, \lambda^2, \ldots \rangle$ we have that $\gamma \lambda^n \in J$ for some
	\[
	n \leq \frac{2\pi}{\length{J}}
	\cdot 2 \pi
	\left(
	\frac{2\pi}{\length{J}}
	\right)^{(\size{\lambda} + \size{1})^C}
	=
	2 \pi
	\left(
	\frac{2\pi}{\length{J}}
	\right)^{(\size{\lambda} + \size{1})^C+1}.
	\]
	Finally, choosing $D$ such that $(\size{\lambda} + \size{1})^C+1 \leq \size{\lambda}^D$ yields the desired result.
\end{proof}

\section{The LTL Model-Checking Problem}

Suppose we are given a matrix $M \in \rationals^{3 \times 3}$, a point $s \in \rationals^3$ and an LTL formula $\varphi$ over semialgebraic predicates $T_1, \ldots, T_m \subseteq \reals^3$ as an input. 
We associate with each $T_i$ an atomic proposition $P_i$ which evaluates to \emph{true} at time $n$ in case $M^ns \in T_i$. 
Hence we associate an $\omega$-word $w$ over $2^{\{P_1, \ldots, P_m\}}$ with the orbit $\langle s, Ms, M^2s, \ldots \rangle$ in the standard manner. The LTL Model-Checking Problem is then to decide whether $w \models \varphi$. 

We assume that $\varphi$ is given in the form
\[
\varphi := T \; | \; \varphi \land \varphi \; | \; \varphi \lor \varphi \; | \; \varphi \Until \varphi \; | \; \varphi \Release \varphi \; | \; \Next \varphi
\]
where $T$ is an atomic semialgebraic set described via a single polynomial inequality of the form $p(x) > 0$ or $p(x) \geq 0$. Observe that $\varphi$ does not contain the $\lnot$ operator: an arbitrary formula $\psi$ over semialgebraic sets (defined by a Boolean combination of inequalities of the type $p(x) > 0$ or $p(x) \geq 0$) can be translated into an equivalent formula in this form by first translating $\psi$ into negation-normal form and then replacing $\lnot p(x) \geq 0$ with $-p(x) > 0$ and $\lnot p(x) > 0$ with $-p(x) \geq 0$. This translation incurs at most a linear blowup in size.

Throughout the paper we assume that the polynomials $p$ defining the atomic predicates are given as a list of coefficients including (possibly many) zeros. Hence it is always the case that $\size{p} \geq \deg{(p)}$. 

We analyse the problem based on the eigenvalues of $M$. Our main
result is the following. 
\begin{theorem}
	Given $M \in \rationals^{3 \times 3}$, $s \in \rationals^3$ and $\varphi$, the LTL Model-Checking Problem is decidable in $\mathop{\mathbf{EXPSPACE}}$ in $\size{M} + \size{s} + \size{\varphi}$.
\end{theorem}

\section{When not all three eigenvalues are real}
We first consider the case in which $M$ has complex eigenvalues $\lambda$, $\overline{\lambda}$ and a real eigenvalue $\rho$.
Moreover, we assume that $\lambda^k \notin \reals $ for all $k \in \naturals$, i.e. $\gamma = \frac{\lambda}{\length{\lambda}}$ is not a root of unity. This case requires by far the most detailed analysis. However, the final model-checking algorithm is quite simple in that it does not involve any non-trivial manipulations of algebraic numbers or semialgebraic sets.

\subsection{Preliminary analysis}

In this section we introduce normalised expressions in order to study the set of all values of $n$ for which the term $M^n s$ of the orbit is in an atomic semialgebraic set $T$. The treatment here mostly mirrors that in \cite{ShaullJournal}.

Since $M$ is assumed to have three distinct eigenvalues, we can diagonalise $M = PDP^{-1}$ where
$ 
D = 
\begin{bmatrix}
\lambda & 0 & 0 \\
0 & \overline{\lambda} & 0 \\
0 & 0 & \rho
\end{bmatrix}
$.
Observe that $M^n s = P D^n P^{-1} s$ and $P, P^{-1}$ contain algebraic entries of constant degree and height polynomial in $\size{M}$. Hence we can write
\[
M^n s = 
\begin{bmatrix}
a_1 \lambda^n +  \overline{a_1} \overline{\lambda}^n + c_1 \rho^n\\
a_2 \lambda^n +  \overline{a_2} \overline{\lambda}^n + c_2 \rho^n \\
a_3 \lambda^n +  \overline{a_3} \overline{\lambda}^n + c_3 \rho^n
\end{bmatrix}
\]
where $a_1, a_2, a_3$ and $c_1, c_2, c_3$ are all algebraic numbers with fixed degree and description length polynomial in $\size{M} + \size{s}$. 

Let $T = \{x \in \reals^3 : p(x) \sim 0\}$, ${\sim} \in \{>, \geq\}$ be an atomic semialgebraic set. We study the set $\Z{T} = \{n\geq0 : M^ns \in T\} = \{n \geq0 : p(M^ns) \sim 0\}$. In Appendix C we show that $p(M^n s)$ can be written as  
\begin{equation}
\sum_{0 \leq p_1, p_2, p_3 \leq deg(p)} 
\alpha_{p_1, p_2, p_3} \lambda^{np_1} \overline{\lambda}^{np_2} \rho^{np_3} 
+
\overline{\alpha_{p_1, p_2, p_3}} \lambda^{np_1} \overline{\lambda}^{np_2} \rho^{np_3}
\end{equation}
where each $\alpha_{p_1, p_2, p_3}$ is an algebraic number of degree polynomial and height exponential in $\size{M} + \size{s} + \size{p}$.
If all the coefficients $\alpha_{p_1, p_2, p_3}$ above are 0 then $p(M^ns) = 0$ for all $n$, and hence the orbit is either always or never in the semialgebraic set $T$. In this case, we can replace $T$ in any LTL formula with $\mathbf{\mathop{true}}$ or $\mathbf{\mathop{false}}$. Otherwise, let $\Lambda = \max \{\length{\lambda^{p_1} \overline{\lambda}^{p_2} \rho^{p_3}}  : \alpha_{p_1, p_2, p_3} \neq 0 \}$.
Dividing the expression in (1) by $\Lambda$ we obtain that $p(M^n s) \sim 0$ if and only if
\[
\sum_{m=0}^{k} \beta_m \gamma^{nm} + \overline{\beta_m} \overline{\gamma}^{nm} + r(n) \sim 0
\]
where
\begin{itemize}
	\item $\gamma = \frac{\lambda}{\length{\lambda}}$ with degree at most 12 (Appendix~B) and height polynomial in $\size{M}$,
	\item $k \leq \deg(p)$,
	\item $r(n) = \sum_{l=1}^{k'} \chi_l \mu_l^{n} + \overline{\chi_l} \overline{\mu_l}^{n}$
	with $\length{\mu_l} < 1$ for every $1 \leq l \leq k'$,
	\item and all the coefficients and exponents $\beta_m$, $0 \leq m \leq k$ and $\chi_l$, $\mu_l$, $1 \leq l \leq k'$ are algebraic.
\end{itemize}

We refer to $e(n) = \sum_{m=0}^{k} \beta_m \gamma^{nm} + \overline{\beta_m} \overline{\gamma}^{nm} + r(n)$ as the \textit{normalised expression corresponding to $T$}, and denote the bit-length of its syntactic representation by $\size{e}$, which can again be shown to be of size polynomial in  $\size{M} + \size{s} + \size{p}$.

\subsection{On visiting atomic semialgebraic sets}
Recall that we defined, for an atomic semialgebraic set $T$, a matrix $M$ and a starting point $s$, $\Z{T} = \{n\geq0 : M^n s \in T\}$. We now study the structure of $\mathcal{Z}(T)$ and show how to compute a useful finite representation for it.

In this section, let $\size{\mathcal{I}_T} = \size{M} + \size{s} + \size{p}$, where $p$ is the polynomial defining $T$.	
Let $e(n) = \sum_{m=0}^{k} \beta_m \gamma^{nm} + \overline{\beta_m} \overline{\gamma}^{nm} + r(n)$ be the normalised expression corresponding to $T$. We call the function $f(z) = \sum_{m=0}^{k} \beta_m z^m + \overline{\beta_m} \overline{z}^m$, $f : \mathbb{C} \rightarrow \reals$ the \textit{dominant function corresponding to $T$}. Observe that $e(n) = f(\gamma^n) + r(n)$.

From \cite{ShaullJournal} we know that $f$ has at most $4k$ zeros in the unit circle $\torus$, which are algebraic numbers with description length polynomial in $\size{f}$, and that
\begin{lemma}
	There exists $N = {2^{\size{e}}} ^ {O(1)}$  such that for all $n > N$,
	\begin{itemize}
		\item $f(\gamma^n) \neq 0$, and
		\item $f(\gamma^n) > 0$ iff $f(\gamma^n) + r(n) > 0$ iff $M^n s \in T$.
	\end{itemize}
\end{lemma}
Since $f$ is a continuous real-valued function, it maintains its sign between its (at most $4k$) roots on the unit circle. Recalling that $\size{e} = \size{\mathcal{I}_T}^{O(1)}$, we rephrase Lemma 4 as follows:
\begin{theorem}
	Let $T$ be an atomic semialgebraic set. There exist $N = {2^{\size{\mathcal{I}_T}}} ^ {O(1)}$ and $J \subseteq \torus$ that is a union of finitely many open intervals such that for $n > N$, $M^ns \in T$ if and only if $\gamma^n \in J$. 
\end{theorem}
Such $J$ can be written uniquely as a disjoint union of open arcs in $\torus$.  We refer to such intervals as the \emph{component intervals} or \emph{components} of $J$.
Observe that the endpoints of components of $J$ are roots of $f$. Recall that $\gamma$ is not a root of unity, and hence by Kronecker's theorem, the sequence $(\gamma^i)_{i \in \naturals}$ is dense in $\torus$. It follows that the sequence $(\gamma^{N+i})_{i \in \naturals}$ is likewise dense in $\torus$, and we obtain that $J$ is unique, in the sense of being independent from any $N$ that satisfies the conclusion of Theorem 5. Hence we refer to $J$ as the \emph{finite union of (open) intervals corresponding to $T$}.

\subsection{$\mathcal{Z}(\varphi)$ for general $\varphi$}
We now study the set $\mathcal{Z}(\varphi) =\{n \geq 0 : \langle M^ns, M^{n+1}s, \ldots \rangle \models \varphi\}$, i.e. the set of all suffixes of the original orbit $\langle s, Ms, M^2s, \ldots \rangle$ that satisfy $\varphi$, for an arbitrary LTL formula $\varphi$. We extend Theorem 5 by showing that $\mathcal{Z}(\varphi)$ also has a corresponding union of finitely many open intervals that can be effectively computed from the finite unions of open intervals corresponding to its subformulas.

In this section, let $\Ib = \size{M}+\size{s} + \sum_{i=1}^m
\size{T_i}$ where $T_1, T_2, \ldots, T_m$ are atomic predicates appearing in some formula $\varphi$. Intuitively, $\left\Vert \mathcal{I}_b \right\Vert$ is the ``basic length'' of the input that doesn't account for the structure of $\varphi$. Our main result is the following:

\begin{theorem}
	Let $\varphi$ be an LTL formula. There exists $N = {2^{\size{\mathcal{I}_b}}}^{O(1)}$ and a finite union of open intervals $J_\varphi \subseteq \torus$ such that for all $n > N$, $n \in \mathcal{Z}(\varphi)$ if and only if $\gamma^n \in J_\varphi$.
\end{theorem}
To prove this, we will combine Theorem 5 with the following result:

\begin{theorem} \label{inductiveConstruction}
	Let semialgebraic sets $T_1, T_2, \ldots, T_m$, time step $N$ and finite unions of open intervals $J_1, J_2, \ldots, J_m \subseteq \torus$ be such that for all $1 \leq i \leq m$ and time steps $n> N$, $n\in \mathcal{Z}(T_i)$ if and only if $\gamma^n\in J_i$. Then for every LTL formula $\varphi$ over $T_1, T_2, \ldots, T_m$ there exists a finite union of open intervals $J_\varphi \subseteq \torus$ such that for all $n> N$, $n\in \mathcal{Z}(\varphi)$ if and only if $\gamma^n\in J_\varphi$. Moreover, such $J_\varphi$ is unique.
\end{theorem}

\begin{proof}
	The uniqueness of $J_\varphi$, if such a finite union of open sets exists, can be established using the same topological argument as the one used in the uniqueness result accompanying Theorem 5. 
	In order to prove existence of $J_\varphi$ with the desirable properties we proceed by induction on the structure of $\varphi$. If $\varphi = T_i$, then $J_\varphi=J_i$ by assumption.
	
	Next, let $J_{\varphi_1}$ and $J_{\varphi_2}$ be the finite unions of open intervals corresponding to $\varphi_1$ and $\varphi_2$, respectively. Recall from Section 3 that we can assume $\varphi$ does not contain the $\lnot$ operator.
	\begin{enumerate}
		\item Suppose $\varphi = \varphi_1 \lor \varphi_2$. Then $J_\varphi = J_{\varphi_1}\cup J_{\varphi_2}$.
		
		\item Similarly, if $\varphi = \varphi_1 \land \varphi_2$ then $J_\varphi = J_{\varphi_1}\cap J_{\varphi_2}$.
		
		\item Consider $\varphi = \Next \varphi_1$. Suppose $n> N$. Then
		\begin{equation*}
		\begin{split}
		n\in \mathcal{Z}(\Next \varphi_1) 
		&\iff
		n+1 \in   \mathcal{Z}(\varphi_1) \\
		& \iff \gamma^{n+1} \in J_{\varphi_1}\\
		& \iff \gamma^n \in \gamma^{-1} J_{\varphi_1}.		
		\end{split}
		\end{equation*} 
		The first equivalence follows from the semantics of the $\Next$ operator, and the second from the fact that	 $n+1 > N$.
		Hence $J_\varphi = \gamma^{-1} J_{\varphi_1}$.
		
		\item The main difficulty lies in analysing the case $\varphi =  \varphi_1 \Until \varphi_2$. If $J_{\varphi_2}$ is empty, then $J_\varphi$ will be empty too. 
		Now suppose $J_{\varphi_2}$ is not empty, and let $l$ be length of a maximal interval in $J_{\varphi_2}$. 
		Using Lemma 2 we can effectively compute a bound 
		\[
		b = b(\varphi_2) = 2\pi\left(
		\frac{2\pi}{l}
		\right)^{\size{\gamma}^D}
		\] 
		such that $\gamma^n$ returns to $J_{\varphi_2}$ after at most $b$ time steps---that is, for every $n > N$, there exists $0 \leq \Delta \leq b$ such that $\gamma^{n+\Delta} \in J_{\varphi_2}$. 
		Thus we have that for $n > N$, 
		\begin{equation*}
		\begin{split}
		n \in \mathcal{Z}(\varphi_1 \Until \varphi_2)
		& \iff \exists \Delta  \geq 0. \text{ } n+\Delta \in \mathcal{Z}(\varphi_2) \land \forall m \in [n, n+\Delta). \text{ } m \in \mathcal{Z}(\varphi_1)  \\
		& \iff \exists \Delta \in [0,b] . \text{ } n+\Delta \in \mathcal{Z}(\varphi_2) \land \forall m \in [n, n+\Delta). \text{ } m \in \mathcal{Z}(\varphi_1)  \\
		& \iff \bigvee_{\Delta = 0}^{b} \left(
		n+\Delta \in \mathcal{Z}(\varphi_2) \land \bigwedge_{m
			= 0}^{\Delta-1} n+m \in 
		\mathcal{Z}(\varphi_1) \right) .
		\end{split}
		\end{equation*} 
		By the induction hypothesis, $n+\Delta \in \mathcal{Z}(\varphi_2) \iff \gamma^{n+\Delta} \in J_{\varphi_2}$, which is equivalent to $\gamma^n \in \gamma^{-\Delta}J_{\varphi_2}$. Similarly, $n+m \in \mathcal{Z}(\varphi_1) \iff \gamma^n \in \gamma^{-m}J_{\varphi_1}$. Hence we obtain that  
		\[
		n \in \mathcal{Z}(\varphi_1 \Until \varphi_2) 
		\iff
		\gamma^n \in 
		\bigvee_{\Delta = 0}^{b}
		\left(
		\gamma^{-\Delta}J_{\varphi_2} \land
		\bigwedge_{m = 0}^{\Delta-1} \gamma^{-m}J_{\varphi_1}
		\right)
		\]
		and therefore, the union of open intervals corresponding to $\varphi$ is 
		\[
		J_\varphi = \bigcup_{\Delta = 0}^{b}
		\left(
		\gamma^{-\Delta}J_{\varphi_2} \cap
		\bigcap_{m = 0}^{\Delta-1} \gamma^{-m}J_{\varphi_1}
		\right)
		.
		\]

		\item Finally, suppose $\varphi = \varphi_1 \Release \varphi_2$. If $J_{\varphi_1 \land \varphi_2} = J_{\varphi_1} \cap J_{\varphi_2}$ is empty, then $J_\varphi = \torus$ if $J_{\varphi_2} = \torus$ and $J_\varphi$ is empty otherwise. If, on the other hand, $J_{\varphi_1 \land \varphi}$ is not empty, then $J_\varphi = J_{\varphi_2 \Until (\varphi_1 \land \varphi_2)}$ which can be computed using the preceding analysis.  
		
		
	\end{enumerate}
\end{proof}
From the construction described above we can observe that the endpoints of components of $J_{\varphi}$ come from those of its immediate subformulas. For example, the endpoints in $J_{\varphi_1 \Until \varphi_2}$ are all endpoints of either $\varphi_1$ or $\varphi_1$ which have been multiplied by $\gamma^{-1}$ for at most $b(\varphi_2)$ steps. In general, the endpoints of component intervals in $J_{\varphi}$ are of the form $\gamma^{-n} z$ where $n$ is an integer and $z$ is a zero of a dominant function (as defined in Section 4.2) corresponding to some semialgebraic target set $T_i$ appearing in $\varphi$.

To prove Theorem 6, recall from Theorem 5 that we already know that for each atomic $T_i$, there exist $N_i = 2^{(\size{M} + \size{s} + \size{T_i})^{O(1)}}$ and $J_i$ such that for $n> N$, $n\in \mathcal{Z}(T_i)$ if and only if $\gamma^n\in J_i$. Taking $N = \max_{1 \leq i \leq m} N_i = {2^{\size{\mathcal{I}_b}}}^{O(1)}$ we obtain the desired result.

\subsection{Analysing the inductive construction of $\Z{T}$ quantitatively}

We now study how small the component intervals in the set $J_\varphi$ corresponding to a formula $\varphi$ can be. Our aim with this analysis is to be able to bound the \emph{return time} $\mathop{T}(\varphi)$ of $\varphi$ with respect to the orbit $\pi = \langle s, Ms, M^2s, \ldots \rangle$, defined as 
\[
T(\varphi) = \sup \{ t_2-t_1 \mid t_1, t_2 \in \naturals \land \pi[t_2, \infty) \vDash \varphi \land \pi[t, \infty) \nvDash \varphi \text{ for every } t_1 \leq t < t_2 \}.
\] 
Informally, $T(\varphi)$ denotes the longest time $\varphi$ remains false in $\pi = \langle s, Ms, M^2s, \ldots \rangle$ before becoming true. 

Recall from Section 4.3 that the endpoints of intervals in $J_\varphi$ are of the form $\gamma^{-n}z$ for some $z$ that is a root of a dominant function corresponding to an atomic predicate $T_i$ appearing in $\varphi$. Hence for an endpoint $u \in \torus$ we define the \emph{retraction depth} of $u$ to be the smallest integer $n$ such that $u = \gamma^{-n}z$ for such $z$.
Next, for an LTL formula $\varphi$ over $T_1, \ldots, T_m$, define
\begin{itemize}
	\item $d(\varphi)$ to be the length of a smallest maximal interval  in the finite union $J_\varphi$ of intervals corresponding to $\varphi$,
	\item $R(\varphi)$ to be the maximum of retraction depths of endpoints in $J_\varphi$, called the \emph{retraction depth} of $\varphi$, and
	\item $D(\varphi)$ denote the maximum nesting depth of temporal operators in $\varphi$, with atomic formulas having depth $0$. 
\end{itemize}
Further, throughout this section let $\left\Vert \mathcal{I}_b \right\Vert = \size{M} + \size{s} + \sum_{i=1}^m \size{T_i}$ and $\I = \size{M} + \size{s} + \size{\varphi}$ where $T_1, \ldots, T_m$ are the atomic semialgebraic sets appearing in the input formula $\varphi$. We link the quantities defined above by analysing the inductive construction described in Theorem \ref{inductiveConstruction}.
\begin{lemma}
	For every $\varphi$, $d(\varphi) \geq \frac{1}{(R(\varphi)+2)^{\size{\mathcal{I}_b}^{O(1)}}}$.
\end{lemma}
\begin{proof}
	We proceed by bounding how close two endpoints of an interval in $J_{\varphi}$ can be given the retraction depth $R(\varphi)$ of $J_{\varphi}$. 
	Let $z_1, z_2$ be roots of dominant functions corresponding to $T_1, T_2$ that appear in $\varphi$ and $\gamma^{-n_1}z_1, \gamma^{-n_2}z_2$ two endpoints of component intervals in $J_{\varphi}$. We show that $\size{\gamma^{-n_1}z_1 - \gamma^{-n_2}z_2} \geq \frac{1}{(N+2)^{\size{\mathcal{I}_b}^{O(1)}}}$ where $N = \length{n_1-n_2}$. The statement of the lemma then follows from the fact that $N \leq R(\varphi)$ by definition.

	For simplicity assume $n_1 \geq n_2$. Recall that the roots $z_1, z_2$  have degree $\Ib^{O(1)}$ and height $2^{\Ib^{O(1)}}$ whereas $\gamma$ has degree at most $12$ and height $\Ib^{O(1)}$. Next, observe that
	\[
	\size{\gamma^{-n_1}z_1 - \gamma^{-n_2}z_2} = \size{z_1 - \gamma^{n_1-n_2}z_2} = \size{\frac{z_1}{z_2} - \gamma^{n_1-n_2}} = \size{z' - \gamma^{n_1-n_2}}
	\]
	where $z' \in \torus$ is also of degree $\Ib^{O(1)}$ and height $2^{\Ib^{O(1)}}$. 
	\begin{itemize}
		\item If $n_1-n_2 < 2$, then we use the root separation bound given in Section 2.1 to obtain that $\size{z_1 - \gamma^{n_1-n_2}z_2} \geq \frac{1}{2^{\Ib^{O(1)}}}$. This can done by separating the roots of the product polynomial $p_1 p_2$ where $p_1$, $p_2$ are polynomials that have $z'$ and $\gamma^{n_1-n_2}$ as roots. Observe that $p_1 p_2$ itself is also of degree $\Ib^{O(1)}$ and height $2^{\Ib^{O(1)}}$.
		\item If $n_1-n_2 \geq 2$, then by a direct application of Lemma 1 we obtain that
		\[
		\size{z_1 - \gamma^{n_1-n_2}z_2} \geq \frac{1}{(n_1-n_2)^{(\size{\gamma} + \size{z'})^C}} = \frac{1}{(n_1-n_2)^{\Ib^{O(1)}}}.
		\]
	\end{itemize}
	Combining the two bounds yields the desired result. 
\end{proof}

We now move onto analysing the retraction depth of $\varphi$. If $\varphi$ does not contain any temporal operators then all the endpoints in $J_{\varphi}$ are roots of dominant functions themselves. Hence, by definition, $R(\varphi) = 0$. For general $\varphi$, on the other hand, we have the following result.

\begin{lemma}
	For every formula $\varphi$ with temporal operator depth $D(\varphi) > 0$ there exist $k \leq D(\varphi)$ formulas $\varphi_1, \ldots, \varphi_k$ over the same atomic predicates as $\varphi$ such that $D(\varphi_1) < D(\varphi_2) < \cdots < D(\varphi_k) < D(\varphi)$ and  
	\[
	R(\varphi) \leq k + 2 \pi \sum_{i=1}^{k}
	\left(\frac{2\pi}{d(\varphi_i)}\right)^
	{{\size{\gamma}^D}}.
	\]
\end{lemma}
\begin{proof}
	Let $\varphi$ be a formula with $D(\varphi) > 0$. We first show that there exist a subformula $\varphi'$ of $\varphi$ with smaller temporal operator depth and a formula $\varphi''$ (possibly not a subformula of $\varphi$) with $D(\varphi'') < D(\varphi)$ such that
	\[
	R(\varphi) \leq 1 + R(\varphi') + 2 \pi \left(\frac{2\pi}{d(\varphi'')}\right)^
	{{\size{\gamma}^D}}.
	\]
	The statement of the lemma then follows by repeatedly applying the inequality to $R(\varphi')$ at most $D(\varphi)$ times. We proceed by a case analysis on the structure of $\varphi$.
	\begin{enumerate}
		\item If $\varphi = \varphi_1 \land \varphi_2$ or $\varphi = \varphi_1 \lor \varphi_2$, then $R(\varphi) \leq \max \{R(\varphi_1), R(\varphi_2)\}$. Hence we can take $\varphi'$ to be the immediate subformula with the larger retraction depth and $\varphi''$ to be any subformula of $\varphi$;
		
		\item If $\varphi = \Next \varphi_1$, then $R(\varphi) \leq 1 + R(\varphi')$ for a smaller subformula $\varphi'$ (namely $\varphi_1$);
		
		\item If $\varphi = \varphi_1 \Until \varphi_2$, and $J_{\varphi_2}$ is empty, then so is $J_{\varphi}$ and $R(\varphi) = 0$. If $J_{\varphi_2}$ is not empty, then the inductive construction shows that the endpoints of $J_\varphi$ are all endpoints of $J_{\varphi_1}$ or $J_{\varphi_2}$ multiplied by $\gamma^{-1}$ for at most $b = b(\varphi_2) = 2 \pi \left(
		\frac{2\pi}{d(\varphi_2)}
		\right)^{\size{\gamma}^D}$ steps (Theorem \ref{inductiveConstruction}, case 4). 
		Hence $R(\varphi) \leq \max \{R(\varphi_1), R(\varphi_2)\} + b(\varphi_2)$ and we can take
		$\varphi'$ to be the immediate subformula of $\varphi$ with larger retraction depth and $\varphi''$ to be $\varphi_2$ (which indeed does have a smaller temporal operator depth and also happens to be a subformula of $\varphi$).
		
		\item Finally, suppose $\varphi = \varphi_1 \Release \varphi_2$. The only non-trivial case arises from non-empty $J_{\varphi_1 \land \varphi_2}$, where the endpoints of $J_{\varphi}$ are all endpoints of $J_{\varphi_1 \land \varphi_2}$ or $J_{\varphi_2}$ multiplied by $\gamma^{-1}$ for at most $b(\varphi_1 \land \varphi_2)$ steps. In this case, similarly to the above, $R(\varphi) \leq R(\varphi') + 2 \pi \left(
		\frac{2\pi}{d(\varphi'')}\right)^{\size{\gamma} ^D}$ where $\varphi'$ is the immediate subformula of $\varphi$ with smaller temporal operator depth and $\varphi''$ is a formula with smaller temporal operator depth (namely, $\varphi_1 \land \varphi_2$, which is not a subformula of $\varphi$).
	\end{enumerate}
\end{proof}

We now combine Lemmas 8 and 9 in the following way. Let $q$ be a polynomial such that ${\size{\gamma}^D} \leq q(\Ib)$ and $d(\varphi) \geq \frac{1}{(R(\varphi)+2)^{q(\Ib)}}$. We obtain that for every $\varphi$ with temporal operator depth $D(\varphi) > 0$ there exist $k \leq D(\varphi)$ formulas $\varphi_1, \ldots, \varphi_k$ with $D(\varphi_1) < D(\varphi_2) < \cdots < D(\varphi_k) < D(\varphi)$ such that
\[
d(\varphi) \geq \frac{1}{(R(\varphi)+2)^{q(\Ib)}} \geq 
\frac{1}{
	\left(
	k + 2 + 2 \pi \sum_{i=1}^{k} 	 \left(\frac{2\pi}{d(\varphi_i)}\right)^
	{q(\Ib)}
	\right)^{q(\Ib)}
}
.
\]
In Appendix D we analyse this recursive relation and show the following.
\begin{theorem}
	For any LTL formula $\varphi$, $d(\varphi) = \frac{1}{
		2^  {2^ {\left(\Ib + D(\varphi)\right)^{O(1)}} } 
	}$. In particular, $d(\varphi) = \frac{1}{
		2^  {2^ {\I^{O(1)}} } 
	}$, where $\I = \size{M}+\size{s}+\size{\varphi}$.
\end{theorem}
The interpretation is that all intervals in $J_\varphi$ have length bounded below by the reciprocal of quantity whose magnitude is
doubly exponential in the length of the input.
In particular, we can compute a uniform lower bound that only depends on the encoding length of the atomic predicates and the depth of the temporal operators (in addition to $\size{M}$ and $\size{s}$) and not on the structure of $\varphi$. 

We are now in a position to apply our quantitative analysis to the model-checking problem via the return time, as discussed at the beginning of this section.
\begin{theorem}
	The return time $T(\varphi)$ of any LTL formula $\varphi$ with respect to an orbit $\langle s, Ms, M^2s, \ldots \rangle$ is $2^  {2^ {\left(\Ib + D(\varphi)\right)^{O(1)}} }$. In particular, $T(\varphi) = 2^  {2^ {\I^{O(1)}} }$.
\end{theorem}
\begin{proof}
	Let $N = 2^{\Ib^{O(1)}}$ be the time after which whether the suffix $\langle M^ns, M^{n+1}s, \ldots \rangle \models \varphi$ depends only on $\gamma^n$, as described in Theorem 7. Applying Lemma 2 to Theorem 10 we obtain that the return time $T'(\varphi)$ of $\varphi$ with respect to $\langle M^{N+1}s, M^{N+2}s, \ldots \rangle$ is $2^{2^{\left(\Ib + D(\varphi)\right)^{O(1)}}}$. Hence the return time with respect to the original orbit is at most $N + T'(\varphi) = 2^{2^{\left(\Ib + D(\varphi)\right)^{O(1)}}}$. 
\end{proof}
We will use this result in Section 6 to construct, given an input formula, an equivalent formula (with respect to the given orbit) that only has bounded quantifiers and then proceed to solve the resulting finitary model-checking problem.

\section{The remaining cases}
Now suppose $M$ has three real eigenvalues or eigenvalues $\lambda, \overline{\lambda}, \rho$ with $\gamma = \frac{\lambda}{\length{\lambda}}$ a root of unity. In Appendices A and B we show that in both cases, for an atomic semialgebraic target~$T$, $\mathcal{Z}(T) =\{n \geq 0 : \langle M^ns, M^{n+1}s, \ldots \rangle \in T\}$ is a semilinear set for which an explicit representation can be computed. In particular,
\begin{theorem}
	Given a semialgebraic set $T$, a square matrix $M \in \rationals^{3 \times 3}$ with three real eigenvalues, and a starting point $s \in \rationals$, there exists an integer $N = {2 ^ {\size{\mathcal{I}_T}}}^{O(1)}$ and a computable $X \subseteq \{0,1\}$ such that for all $n > N$, $M^n s \in S$ if and only if $n \bmod 2 \in X$.
\end{theorem}

\begin{theorem}
	Given a semialgebraic set $T$, a square matrix $M \in \rationals^{3 \times 3}$ with eigenvalues $\lambda, \overline{\lambda}, \rho$ where $\gamma = \frac{\lambda}{\length{\lambda}}$ is a root of unity, and a starting point $s \in \rationals$, there exists an integer $N = {2 ^ {\size{\mathcal{I}_T}}}^{O(1)}$ and a computable $X \subseteq \{0,1, \ldots, 287\}$ such that for all $n > N$, $M^n s \in S$ if and only if $n \bmod 288 \in X$.
\end{theorem}
Here once again $\size{\mathcal{I}_T} = \size{p} + \size{M} + \size{s}$, where $p$ is the polynomial defining $T$. In the next section we discuss how to utilise Theorems 12 and 13 in order to obtain a decision procedure for the relevant cases of the LTL Model-Checking Problem.

\section{Model-checking algorithm and its complexity}

\begin{lstlisting}[caption={Recursive model-checking algorithm for formulas with only bounded temporal operators.},label=list:8-6,captionpos=t,float,abovecaptionskip=-\medskipamount]
ModelCheck(formula F, starting point n)
case F = Until(F1, F2, upper bound B):
for i=0 to B do
if ModelCheck(F2, n+i) return true
if not ModelCheck(F1, n+i) return false
return false
case F = Release(F1, F2, upper bound B):
for i=0 to B do
if not ModelCheck(F2, n+i) return false
if ModelCheck(F1, n+i) return true
return true	
case F = Next(F1)
return ModelCheck(F1, n+1)	 
case F = And(F1, F2):
l = ModelCheck(F1, n)
r = ModelCheck(F2, n)
return l and r
case F = Or(F1, F2):
l = ModelCheck(F1, n)
r = ModelCheck(F2, n)
return l or r
case F = atomic semialgebraic T:
return Oracle(T, n)
\end{lstlisting}

In this section we summarize our algorithmic contribution. Suppose we are given $M \in \rationals^{3 \times 3}$, $s \in \rationals^3$ and an LTL formula $\varphi$ over semialgebraic $T_1, \ldots, T_m$ as the input. We describe a decision procedure for determining whether $\langle s, Ms, M^2s ,\ldots \rangle \models \varphi$. 

Let us first consider the complexity of determining, for a given $n$, $M \in \rationals^{3 \times 3}$, $s \in \rationals^3$ and a semialgebraic target $T$ defined via $p(x) \sim 0$, whether $M^n s \in T$. Using iterated squaring we can encode the statement $p(M^n s) \sim 0$ into the existential theory of real numbers using a formula of size $O(\size{M}\log n + \size{p} + \size{s})$. If the input is $M$, $s$ and an LTL formula $\varphi$ containing $T$, this can be written as $O(\I + \log n)$. Since the existential theory of real numbers can be decided in polynomial space (see, e.g., \cite{Renegar88}), an oracle for determining whether $M^n s $ is in a target set $T$ can be implemented using space polynomial in $\I + \log n$. 

We now move onto the main algorithm. As the first step, determine whether $M$ has three real eigenvalues $\rho_1, \rho_2, \rho_3$ or two complex eigenvalues $\lambda, \overline{\lambda}$ and a real eigenvalue $\rho$. If the latter is the case, additionally determine whether $\gamma = \frac{\lambda}{\length{\lambda}}$ is a root of unity or not. 

If $M$ only has real eigenvalues or $\gamma$ is a root of unity, we proceed by computing an explicit representation for the semilinear set $\Z{\varphi} =\{n \geq 0 : \langle M^ns, M^{n+1}s, \ldots \rangle \models \varphi\}$. We illustrate how this can be done by using Theorem 13 and repeatedly combining semilinear sets in case where $\gamma$ is a root of unity. In case $M$ has three real eigenvalues the same procedure can be applied to Theorem 12. 

We first compute an explicit representation for $\Z{T_i} =\{n \geq 0 : \langle M^ns, M^{n+1}s, \ldots \rangle \in T_i\}$ for each atomic $T_i$ in $\varphi$. 
To this end, we compute the value of $N_i = 2^{{\size{\mathcal{I}_{T_i}}}^{O(1)}}$ described in Section 5  for each $1 \leq i \leq m$ and then take the maximum $N = \max_{1 \leq i \leq m} N_i$. Next we determine $F_i = \{n \leq N : M^n s \in T_i\}$ and compute $X_i \subseteq \{0,1, \ldots, 287\}$ such that for $n > N$, $M^n s \in T_i$ if and only if $n \bmod 288 \in X$. These sets can be determined by making queries to the oracle of the form $M^{n}s \stackrel{?}{\in} T_i$ for $0 \leq n \leq N +288$, requiring $2^{\I^{O(1)}}$ space in total. Finally, from sets $F_i$, $X_i$ for $1 \leq i \leq m$ we can construct, for arbitrary formula $\varphi$, sets $F$ and $X$ such that for all $n \leq N$, $\langle M^ns, M^{n+1}s, \ldots \rangle \models \varphi$ if and only if $n \in F$ and for all $n > N$, $\langle M^ns, M^{n+1}s, \ldots \rangle \models \varphi$ if and only if $n \bmod 288 \in X$. It only remains to check whether $0 \in F$. Hence we have a decision procedure that is in $\mathop{\mathbf{EXPSPACE}}$ in $\I$.

If, on the other hand, $\gamma$ is not a root of unity, then we proceed by replacing each $\Release$ and $\Until$ operator in $\varphi$ with a bounded one. Suppose $\varphi_1 \Until \varphi_2$ is a subformula of $\varphi$. Using Theorem 11 we can compute an upper bound $B$ on return time $T(\varphi_2)$ of $\varphi_2$ with respect to $\langle s, Ms, M^2s, \ldots \rangle$. We then simply replace $\varphi_1 \Until \varphi_2$ in $\varphi$ with $\varphi_1 \Until^{\leq B} \varphi_2$ (``$\varphi_1$ remains true until $\varphi_2$ is true, and $\varphi_2$ becomes true within the first $B$ steps''), with the justification that at any time step $n$, if the formula $\varphi_2$ remains false for all suffixes $\langle M^{n+\delta} s, M^{n+\delta+1}s, \ldots\rangle$, $0 \leq \delta \leq B$, then $\varphi_2$ will remain false for all $\langle M^{n+\delta} s, M^{n+\delta+1}s, \ldots\rangle$, $\delta \geq 0$. Similarly, for a subformula of the form $\varphi_1 \Release \varphi_2$ we first compute bounds $B_1$ and $B_2$ on the return times $T(\varphi_1 \land \varphi_1)$ and $T(\lnot \varphi_2)$, respectively, and set $B = \max \{B_1, B_2\}$. Observe that $B$ is at most the bound stipulated in Theorem~11 on the return time of $\varphi_1 \Release \varphi_2$ as the latter has higher temporal operator depth. Finally, we replace $\varphi_1 \Release \varphi_2$ with the bounded version $\varphi_1 \Release^{\leq B} \varphi_2$ with the semantics that either $\varphi_1$ successfully releases $\varphi_2$ within the first $B$ steps or $\varphi_2$ remains true for the first $B$ steps. 

We have now reduced the original problem of checking whether the orbit $\langle s, Ms, M^2s, \ldots \rangle $ satisfies $\varphi$ to determining whether it satisfies $\varphi'$ with all operators bounded by at most $2^{2^{\I^{O(1)}}}$ steps. Moreover, note that our algorithm so far does not involve any manipulation of semialgebraic sets or algebraic numbers. In Listing 1 we describe a simple recursive procedure for determining whether a path satisfies such a formula $\varphi'$ with only bounded temporal operators starting from a time step $n$.

To analyse the complexity of our main algorithm, let $B$ be a maximum bound on a temporal operator in $\varphi'$ (i.e. maximum bound on the return time of a subformula of $\varphi$). Observe that during the run of the model-checking algorithm, all calls to the oracle are for time steps $n \leq D(\varphi') B = 2^  {2^ {\I^{O(1)}} }$, where $D(\varphi)$ is the temporal operator depth of $\varphi$ as defined in Section~4.4. Therefore, the total space required by the oracle is 
$  
O\left(\I +\log (D(\varphi) B) \right) = 2^{\I^{O(1)}}.
$
With respect to the oracle, our algorithm operates in $O(D(\varphi) \cdot \log (D(\varphi) B)) = 2^{\I^{O(1)}}$ space as it simply maintains at most $D(\varphi)$-many counters with $D(\varphi)B$ bits. Adding the two space requirements we conclude that our decision procedure lies in $\mathop{\mathbf{EXPSPACE}}$ in $\I$.

\section{Conclusion}
We have given an algorithm to model check an LTL formula on the orbit
of a linear dynamical system in dimension at most 3.  The procedure
reduces the LTL Model-Checking Problem to an equivalent bounded
model-checking problem, which can be solved directly.
The heart of the proof is the effective upper bound, given in Theorem
11, of the so-called return time of an LTL formula on a given orbit.
Establishing this bound requires the use of several number-theoretic
tools.  As we have noted in the introduction, there are formidable
obstacles to generalising this result to matrices of higher
dimensions, since the LTL Model-Checking Problem generalises numerous
longstanging open decision problems on linear dynamical systems.
Another direction for further work is to consider the problem of model
checking MSO, i.e., to generalise the logic.  Here we plan to explore
connections with the respective frameworks of Semenov~\cite{Semenov}
and Rabinovich~\cite{Rabinovitch} on decidable extensions of MSO with
almost periodic predicates.  Finally, in this work we have considered
the unique orbit determined by a fixed starting point.  But many
situations ask to quantify over different orbits, e.g., one could ask
whether there is a neighbourhood of a given point such that all orbits
starting in the neighbourhood satisfy a given LTL
formula---see~\cite{ShaullJournal} and~\cite{AAGT15} for work in this
direction.
 
\appendix
\section{The case of only real eigenvalues}
Let $M \in \rationals^{3 \times 3}$ be a matrix with three real eigenvalues $\rho_1, \rho_2, \rho_3$, and $S \subseteq \reals^3$ a semialgebraic set defined via a polynomial $p$. We proceed in a way to similar to what we did in Section 4.1. Converting $M$ to Jordan normal form, we can write $M = P^{-1}JP$. Depending on the multiplicity of eigenvalues, three scenarios are possible:
\begin{enumerate}
	\item $J =
	\begin{bmatrix}
	\rho_1 & 0 & 0 \\
	0 & \rho_2 & 0 \\
	0 & 0 & \rho_3
	\end{bmatrix} 
	$. We have $M^n = P^{-1}
	\begin{bmatrix}
	\rho_1^n & 0 & 0 \\
	0 & \rho_2^n & 0 \\
	0 & 0 & \rho_3^n
	\end{bmatrix}
	P$.
	\item $J =
	\begin{bmatrix}
	\rho_1 & 1 & 0 \\
	0 & \rho_2 & 0 \\
	0 & 0 & \rho_3
	\end{bmatrix} 
	$ and $\rho_1=\rho_2$. In this case, $M^n = P^{-1}
	\begin{bmatrix}
	\rho_1^n & n \rho_1^{n-1} & 0 \\
	0 & \rho_1^n & 0 \\
	0 & 0 & \rho_3^n
	\end{bmatrix}
	P$.
	\item $J =
	\begin{bmatrix}
	\rho_1 & 1 & 0 \\
	0 & \rho_2 & 1  \\
	0 & 0 & \rho_3
	\end{bmatrix} 
	$ and $\rho_1=\rho_2=\rho_3$. Then
	$M^n = P^{-1}
	\begin{bmatrix}
	\rho_1^n & n \rho_1^{n-1} & \frac{1}{2}n(n-1)\rho_1^{n-2} \\
	0 & \rho_1^n & n \rho_1^{n-1} \\
	0 & 0 & \rho_1^n
	\end{bmatrix}
	P$. 
\end{enumerate}
In all three cases, we can write

\[
M^n s = 
\begin{bmatrix}
A_1(n) \rho_1^n +  B_1(n) \rho_2^n + C_1(n) \rho_3^n\\
A_2(n) \rho_1^n +  B_2(n) \rho_2^n + C_2(n) \rho_3^n \\
A_3(n) \rho_1^n +  B_3(n) \rho_2^n + C_3(n) \rho_3^n
\end{bmatrix}
\]
where $A_i, B_i, C_i$, $1 \leq i \leq 3$ are polynomials (with real algebraic coefficients) of degree constant (at most 3) and height polynomial in $\size{M} + \size{s}$. Here by height of a polynomial we mean $\max_{1 \leq i \leq d} \length{c_i}$ where $c_i$ ranges over the coefficients of the polynomial. This is a straightforward generalization of height from polynomials with integer coefficients to polynomials with real algebraic coefficients. 

Let $p$ be a polynomial with integer coefficients defining a semialgebraic set $T$. By aggregating coefficients of $p(A_1(n) \rho_1^n +  B_1(n) \rho_2^n + C_1(n) \rho_3^n, A_2(n) \rho_1^n +  B_2(n) \rho_2^n + C_2(n) \rho_3^n, A_3(n) \rho_1^n +  B_3(n) \rho_2^n + C_3(n) \rho_3^n)$, we obtain
\begin{equation}
p(M^ns) = \sum_{0 \leq p_1, p_2, p_3 \leq deg(p)} 
\alpha_{p_1, p_2, p_3} (n) \rho_1^{np_1} \rho_2^{np_2} \rho_3^{np_3}
\end{equation}
where each $\alpha_{p_1, p_2, p_3}$ is a polynomial with real algebraic coefficients and degree polynomial and height exponential in $\size{\mathcal{I}_T} = \size{M} + \size{s} + \size{p}$. These size estimates can be proven in a similar way to size estimates established in Appendix C.

We now consider the question of computing a representation for $\Z{S}$ given an expression for $p(M^n s)$ of the form given in Equation~2. 
Let $P = \{\rho_1^{p_1} \rho_2^{p_2} \rho_3^{p_3} : \alpha_{p_1, p_2, p_3} \neq 0\} = \{\sigma_1, \sigma_2, \ldots, \sigma_k\}$ with $\length {\sigma_i} > \length{\sigma_{i+1}}$ for all $1 \leq i < k$, and denote by $\alpha_i$ the non-zero polynomial corresponding to $\sigma_i$ in the equation for $p(M^n s)$ above. Thus
\[
p(M^n s) = \alpha_1 (n) \sigma_1^n + \sum_{i=1}^k \alpha_i(n) \sigma_i^n.
\]
Our goal is to find $N$ large enough such that for all $n > N$,
\[
\length{\sum_{i=2}^k \alpha_i(n) \sigma_i^n } < \length{\alpha_1(n)\sigma_1^n}.
\]
This will alow us to discard all terms except for $\alpha_1 (n) \sigma_1^n$ when determining the sign of $p(M^n s)$ for $n$ sufficiently large. Let $q$ be a polynomial such that $k < q(\I)$, $\deg(\alpha_i) < q(\I)$ and $H(\alpha_i) < 2^{q(\I)}$ for $1 \leq i \leq k$ (note that these are degree and height of polynomials, not algebraic numbers). 
We estimate
\[
\length{\sum_{i=2}^k \alpha_i(n) \sigma_i^n } \leq 
\sum_{i=2}^k \length{\alpha_i(n) \sigma_i^n}  \leq
q(\I) (2n)^{q(\I)} 
\left(1 - \frac{1}{2^{q(\I)}}\right)^n  \length{\sigma_1}^n
\]
for all $n \geq 1$. Therefore, it suffices to find $N$ large enough such that for all $n > N$
\[
q(\I) (2n)^{q(\I)} 
\left(1 - \frac{1}{2^{q(\I)}}\right)^n \leq 1 
\]
which is equivalent to
\[
\log(q(\I)) + \log(2) + q(\I) \log(n) + n \log \left(1 - \frac{1}{2^{q(\I)}}\right) \leq 0.
\]
Since $\log \left(1 - \frac{1}{2^{q(\I)}}\right) \leq - \frac{1}{2^{q(\I)}}$ and $\log(n) \leq \sqrt{n}$, we choose $N$ to be at least the maximum of the largest root of the quadratic
\[
2 + 2q(\I)\sqrt{n} - n \frac{1}{2^{q(\I)}} \leq 0.
\]
and the largest positive root of $\alpha_1$. For such $N$, $N = 2^{\I^{O(1)}}$ and for all $n > N$, whether $M^n s$ is in atomic semialgebraic set $T$ or not only depends on the sign of $\sigma_1^n$. That is, after time $N$, the orbit is in $T$ either all the time, never, or every other time step. Stated differently, there exists $X \subseteq \{0,1\}$ such that for all $n > N$, $M^n s \in T$ if and only if $n \bmod 2 \in X$.


\section{The case where $\gamma$ is a root of unity}
Let $M \in \rationals^{3 \times 3}$ be a matrix with three eigenvalues $\lambda, \overline{\lambda}, \rho$, $\gamma = \frac{\lambda}{\length{\lambda}}$ and suppose $\gamma$ is a root of unity. That is, there exists $d \in \naturals$ such that $\gamma^d = 1$. We first argue that $d \leq 144$. 
Since $\lambda$ has degree $3$, $\gamma = \frac{\lambda}{\sqrt{\lambda \overline{\lambda}}}$ can have degree at most $12$. 
But since $\gamma$ is a $d$th root of unity, we also know that its defining polynomial is the $d$th cyclotomic polynomial $\varphi_d(x)$ whose degree is $\varphi(d)$, where $\varphi$ is the Euler's totient function. 
Therefore, $d$ must be such that $\varphi(d) \leq 12$.
Since $\varphi(d) \geq \sqrt{d}$, we obtain that $d \leq 144$. 

By diagonalising $M$ (see Section 4.1) we can write
\[
M^n s = 
\begin{bmatrix}
a_1 \lambda^n +  \overline{a_1} {\overline{\lambda}}^n + c_1 \rho^n\\
a_2 \lambda^n +  \overline{a_2} {\overline{\lambda}}^n + c_2 \rho^n \\
a_3 \lambda^n +  \overline{a_3} {\overline{\lambda}}^n + c_3 \rho^n
\end{bmatrix}.
\]
where $a_1, c_1, a_2, c_2, a_3, c_3$ are algebraic numbers. Suppose $n = m \mod d$. Since $\lambda = \gamma \length{\lambda}$,
\[
M^n s = 
\begin{bmatrix}
a_1 \length{\lambda}^n \gamma^m +  \overline{a_1} \length{\lambda}^n\overline{\gamma}^m + c_1 \rho^n\\
a_2 \length{\lambda}^n \gamma^m +  \overline{a_2} \length{\lambda}^n\overline{\gamma}^m + c_2 \rho^n\\
a_3 \length{\lambda}^n \gamma^m +  \overline{a_3} \length{\lambda}^n\overline{\gamma}^m + c_3 \rho^n\\
\end{bmatrix}
= \begin{bmatrix}
2\operatorname{Re}(a_1\gamma^m)\length{\lambda}^n + c_1 \rho^n\\
2\operatorname{Re}(a_2\gamma^m)\length{\lambda}^n + c_2 \rho^n\\
2\operatorname{Re}(a_3\gamma^m)\length{\lambda}^n + c_3 \rho^n\\
\end{bmatrix}.
\]
We thus split the sequence $\langle s, Ms, M^2s, \ldots \rangle$ into $d$ subsequences: for $0 \leq m < d$, define
$x^m_i = M^{id+m}s$ for $i \geq 0$. Then we have
\[
x^m_i = 
\begin{bmatrix}
2\operatorname{Re}(a_1\gamma^m)\length{\lambda}^m \length{\lambda^d}^i + c_1 \rho^m (\rho^d)^i\\
2\operatorname{Re}(a_2\gamma^m)\length{\lambda}^m \length{\lambda^d}^i + c_2 \rho^m (\rho^d)^i\\
2\operatorname{Re}(a_3\gamma^m)\length{\lambda}^m \length{\lambda^d}^i + c_3 \rho^m (\rho^d)^i\\
\end{bmatrix}.
\]
Observe that $\operatorname{Re}(a_i\gamma^m)$, $1 \leq i \leq 3$, and $\length{\lambda}^m$ are real algebraic constants of constant degree and height polynomial in $\size{M} + \size{s}$.

Let $T \subseteq \reals^3$ a semialgebraic set defined via the polynomial $p$. We now investigate when the sequence $(x^m_i)_{i \in \naturals}$ is in the $T$. Consider $p(x^m_i)$. By aggregating coefficients, we can write
\[
p(x^m_i) = \sum_{0 \leq p_1, p_2, p_3 \leq deg(p)} 
\alpha_{p_1, p_2, p_3} \rho_1^{np_1} \rho_2^{np_2} 
\]
where $\rho_1 = \length{\lambda^d}$, $\rho_2 = \rho^d$ and $\alpha_{p_1, p_2, p_3}$ are real algebraic numbers of degree polynomial and height exponential in $\size{\mathcal{I}_T} \leq \I = \size{M} + \size{s} + \size{p}$. 
The size estimates can be established as follows: in Appendix C, when aggregating coefficients, we treat $\lambda^n, \overline{\lambda}^n, \rho^n$ as symbols. Hence if we replace $\lambda$ with $\rho_1$, $\overline{\lambda}$ with $\rho_2$, $\rho$ with 0 and the coefficients $a_i, b_i$ with $2\operatorname{Re}(a_i\gamma^m)\length{\lambda}^m$ and $c_i \rho^m$ for $1 \leq  i \leq 3$, the argument remains valid.

Hence we conclude that there exists $N_m = 2^{\I^{O(1)}}$ such that for $n > N$, whether $x^m_n \in T$ depends only on the parity of $n$. Finally, recall that we have split the original sequence into $144$ subsequences. Taking 
\[
N = 144 \max_{0 \leq i < d} N_i = 2^{\I^{O(1)}}
\]
we obtain that there exists a computable subset $X$ of $\{0, 1, \ldots,  287\}$ such that for $n > N$, $M^n s \in T$ iff $n \mod 288 \in X$.

\section{Size estimates for algebraic numbers in Section 4}
In Section 4 we considered the case where $M$ has eigenvalues $\lambda, \overline{\lambda}$ and $\rho$. We showed that 
\[
M^n s = 
\begin{bmatrix}
a_1 \lambda^n +  \overline{a_1} \overline{\lambda}^n + c_1 \rho^n\\
a_2 \lambda^n +  \overline{a_2} \overline{\lambda}^n + c_2 \rho^n \\
a_3 \lambda^n +  \overline{a_3} \overline{\lambda}^n + c_3 \rho^n
\end{bmatrix}
\]
where $a_1, a_2, a_3$ and $c_1, c_2, c_3$ are all algebraic numbers with fixed degree and description length polynomial in $\size{M} + \size{s}$.
In this section we show that given a semialgebraic set $T$ defined via $p(x) \sim 0$, ${\sim} \in \{>, \geq\}$ 
and
\[
p(x_1, x_2, x_3) = \sum_{0 \leq i,j,k \leq \deg(p)} c_{i,j,k} x_1^i x_2^j x_3^k
\]
we can write 
$p(M^n s) =(a_1 \lambda^n +  \overline{a_1} \overline{\lambda}^n + c_1 \rho^n,
a_2 \lambda^n +  \overline{a_2} \overline{\lambda}^n + c_2 \rho^n,
a_3 \lambda^n +  \overline{a_3} \overline{\lambda}^n + c_3 \rho^n)$ as
\begin{equation*}
\sum_{0 \leq p_1, p_2, p_3 \leq deg(p)} 
\alpha_{p_1, p_2, p_3} \lambda^{np_1} \overline{\lambda}^{np_2} \rho^{np_3} 
+
\overline{\alpha_{p_1, p_2, p_3}} \lambda^{np_1} \overline{\lambda}^{np_2} \rho^{np_3}
\end{equation*}
where each $\alpha_{p_1, p_2, p_3}$ has degree polynomial in and height exponential in $\size{M} + \size{s} + \size{p}$.

It suffices to prove this for monomials of the form $x_1^i x_2^j x_3^k$ for some $i$,$j$,$k$. 
To see this, suppose that for each monomial $p_{i,j,k}=x_1^i x_2^j x_3^k$ appearing in $p(x_1, x_2, x_3)$, we can write $(a_1 \lambda^n +  \overline{a_1} \lambda_2^n + c_1 \rho^n)^i
(a_2 \lambda^n +  \overline{a_2} \lambda_2^n + c_2 \rho^n)^j
(a_3 \lambda^n +  \overline{a_3} \lambda_2^n + c_3 \rho^n)^k$  as 
\[
\sum_{0 \leq p_1, p_2, p_3 \leq i+j+k} 
\alpha^{i,j,k}_{p_1, p_2, p_3} \lambda^{np_1} \overline{\lambda}^{np_2} \rho^{np_3} 
+
\overline{\alpha^{i,j,k}_{p_1, p_2, p_3}} \lambda^{np_1} \overline{\lambda}^{np_2} \rho^{np_3}
\]
where each $\alpha^{i,j,k}_{p_1, p_2, p_3}$ has a degree polynomial and height exponential in $\size{\mathcal{I}_{i,j,k}} = \size{M} + \size{s} + \size{p_{i,j,k}}$ (recall that in our encoding scheme, the representation size of $p_{i,j,k}$ is at least $i+j+k$; see Section 3). Clearly, $\size{\mathcal{I}_{i,j,k}} \leq \I$. 
Now observe that
\[
\alpha_{p_1, p_2, p_3} = \sum_{1 \leq i,j,k \leq \deg(p)} c_{i,j,k}\alpha^{i,j,k}_{p_1, p_2, p_3}
\]
for each $\alpha_{p_1, p_2, p_3}$ where $c_{i,j,k}$ is the integer coefficient of $x_1^i x_2^j x_3^k$ in $p(x_1, x_2, x_3)$. Therefore,
\[
\deg(\alpha_{p_1, p_2, p_3}) = \sum_{1 \leq i,j,k \leq \deg(p)} \deg(\alpha^{i,j,k}_{p_1, p_2, p_3}) = \I^{O(1)},
\]
\[
H(\alpha_{p_1, p_2, p_3}) = \sum_{1 \leq i,j,k \leq \deg(p)} c_{i,j,k} H(\alpha^{i,j,k}_{p_1, p_2, p_3}) = 2^{\I^{O(1)}}.
\]

We now move onto proving the result described above for monomials. Let $x_1^i x_2^j x_3^k$ be a monomial. Using 
\begin{itemize}
	\item $\mathbf{i} \mathbf{j} \mathbf{k}$ as a shorthand for $\langle i_1, i_2, i_3, j_1, j_2, j_3, k_1, k_2, k_3\rangle$, and 
	\item and $B(\mathbf{i} \mathbf{j} \mathbf{k})$ for $\binom{i_1+i_2+i_3}{i_1,i_2,i_3}
	\binom{j_1+j_2+j_3}{j_1,j_2,j_3}
	\binom{k_1+k_2+k_3}{k_1,k_2,k_3}$
\end{itemize}
we can write $(a_1 \lambda^n +  \overline{a_1} \overline{\lambda}^n + c_1 \rho^n)^i
(a_2 \lambda^n +  \overline{a_2} \overline{\lambda}^n + c_2 \rho^n)^j
(a_3 \lambda^n +  \overline{a_3} \overline{\lambda}^n + c_3 \rho^n)^k$  as

\[
\sum_{\substack{i_1 + i_2 + i_3 = i \\ j_1 + j_2 + j_3 = j 
		\\ k_1 + k_2 + k_3 = k}}
B(\mathbf{i}\mathbf{j}\mathbf{k}) \cdot  
a_1^{i_1} \overline{a_1}^{i_2} c_1^{i_3} \cdot
a_2^{j_1} \overline{a_2}^{j_2} c_2^{j_3} \cdot
a_3^{k_1} \overline{a_3}^{k_2} c_3^{k_3} \cdot
\lambda^{i_1 + j_1 + k_1} \overline{\lambda}^{i_2 + j_2 + k_2} \rho^{i_3 + j_3 + k_3}. 
\]
We would like to write this expression as
\[
\sum_{0 \leq p_1, p_2, p_3 \leq i+j+k} 
\beta_{p_1, p_2, p_3} \lambda^{np_1} \overline{\lambda}^{np_2} \rho^{np_3} 
+
\overline{\beta_{p_1, p_2, p_3}} \lambda^{np_1} \overline{\lambda}^{np_2} \rho^{np_3}.
\] 
Matching the expressions we obtain
\[
\beta_{p_1, p_2, p_3} = \sum_{\substack{i_1 + k_1 + j_1 = i \\ i_2 + j_2 + k_2 = j \\ i_3 + j_3+ k_3 = k}} 
B(\mathbf{i}\mathbf{j}\mathbf{k}) \cdot  
a_1^{i_1} \overline{a_1}^{i_2} c_1^{i_3} \cdot
a_2^{j_1} \overline{a_2}^{j_2} c_2^{j_3} \cdot
a_3^{k_1} \overline{a_3}^{k_2} c_3^{k_3}.
\]
Observe that each $B(\mathbf{i}\mathbf{j}\mathbf{k})$ is a positive integer and when viewed as an algebraic number has degree 1 and height at most $3^{i+j+k}$. Hence each $B(\mathbf{i}\mathbf{j}\mathbf{k}) \cdot  
a_1^{i_1} \overline{a_1}^{i_2} c_1^{i_3} \cdot
a_2^{j_1} \overline{a_2}^{j_2} c_2^{j_3} \cdot
a_3^{k_1} \overline{a_3}^{k_2} c_3^{k_3}$ is an algebraic number with degree polynomial and height exponential in $\size{\mathcal{I}_{i,j,k}}$. Since each coefficient $\beta_{p_1, p_2, p_3}$ is a sum of polynomially many such numbers, $\beta_{p_1, p_2, p_3}$ also has degree polynomial and height exponential in $\size{\mathcal{I}_{i,j,k}}$.

\section{Bounding the return time of $\varphi$ in Section 4.4}
In this section we prove by induction on the temporal operator depth of $\varphi$ that there exists a polynomial $h$ such that for any LTL formula $\varphi$, $d(\varphi) = \frac{1}{
	2^{2^{h\left(\Ib + D(\varphi)\right)}}
}.$
The base case of $D(\varphi) = 0$ corresponds to $\varphi$ without a temporal operator and is handled in Lemma 8 (as $R(\varphi)=0$ for $\varphi$ without temporal operators). 
	
For the inductive step, let $\varphi$ be an arbitrary formula with $D(\varphi) > 0$, and $\varphi_1, \ldots, \varphi_k$ be the formulas with $D(\varphi_1) < D(\varphi_2) < \cdots < D(\varphi_k) < D(\varphi)$ such that
\[
d(\varphi) \geq 
\frac{1}{
	\left(
	k + 2 + 2 \pi\sum_{i=1}^{k} \left(\frac{2\pi}{d(\varphi_i)}\right)^
	{q(\Ib)}
	\right)^{q(\Ib)}
}
\]
as per Lemma 9. We use this lower bound in the following way:
\begin{equation*}
	\begin{split}
	d(\varphi) \geq \frac{1}{
		2^  {2^ {h\left(\Ib + D(\varphi)\right)} } 
	}
	& \impliedby
	\frac{1}{
		\left(
		k + 2 + 2\pi \sum_{i=1}^{k} \left(\frac{2\pi}{d(\varphi_i)}\right)^
		{q(\Ib)}
		\right)^{q(\Ib)}
	}
	\geq 
	\frac{1}{
		2^  {2^ {h\left(\Ib + D(\varphi)\right)} } 
	}
	\\
	& \impliedby
	\left(
	k + 2 + 2\pi \sum_{i=1}^{k} \left(\frac{2\pi}{d(\varphi_i)}\right)^
	{q(\Ib)} 
	\right)^{q(\Ib)}
	\leq 2^{2^ {h\left(\Ib + D(\varphi)\right)} }
	\end{split}
	\end{equation*} 
By using the induction hypothesis that $d(\varphi_i) \geq \frac{1}{
	2^{2^{h\left(\Ib + D(\varphi_i)\right)}}
} \geq \frac{1}{
	2^{2^{h\left(\Ib + D(\varphi) - 1\right)}}
}$ for each $1 \leq i \leq k$ and few simple upper bounds (e.g. $k \leq \Ib$), we obtain that
\begin{equation*}
d(\varphi) \geq \frac{1}{
	2^  {2^ {h\left(\Ib + D(\varphi_i)\right)} } 
}
\impliedby
\left(
2^  {B \cdot q(\Ib) + 2^ {h\left(\Ib + D(\varphi) - 1\right)} } 
\right)^{q(\Ib)} 
\leq 2^  {2^ {h\left(\Ib + D(\varphi)\right)} }
\end{equation*} 
for some sufficiently large constant $B$ that does not depend on $\Ib$ or any of the (sub)formulas. Hence it suffices to prove
\[
B \cdot q(\Ib)^2 + q(\Ib) 2^ {h\left(\Ib + D(\varphi) - 1\right)} \leq 2^ {h\left(\Ib + D(\varphi)\right)}
\]
which can be achieved by choosing $h\left(\Ib + D(\varphi)\right) = \left(\Ib + D(\varphi)\right)^K$ for $K$ sufficiently large with respect to constants $q$ and $B$. Again, this choice does not depend on the concrete values of $\Ib$ and $D(\varphi)$.


\bibliography{Main}

\begin{thebibliography}{10}

\bibitem{AAGT15}
Manindra Agrawal, S.~Akshay, Blaise Genest, and P.~S. Thiagarajan.
\newblock Approximate verification of the symbolic dynamics of {M}arkov chains.
\newblock {\em J. {ACM}}, 62(1):2:1--2:34, 2015.

\bibitem{AOW17}
S.~Almagor, J.~Ouaknine, and J.~Worrell.
\newblock The {Polytope-Collision Problem}.
\newblock In {\em 44th International Colloquium on Automata, Languages, and
  Programming, {ICALP} 2017, July 10-14, 2017, Warsaw, Poland}, pages
  24:1--24:14, 2017.

\bibitem{AOW19b}
Shaull Almagor, Jo{\"{e}}l Ouaknine, and James Worrell.
\newblock The {Semialgebraic Orbit Problem}.
\newblock In {\em 36th International Symposium on Theoretical Aspects of
  Computer Science, {STACS} 2019, March 13-16, 2019, Berlin, Germany}, pages
  6:1--6:15, 2019.

\bibitem{ShaullJournal}
Shaull Almagor, Jo{\"e}l Ouaknine, and James Worrell.
\newblock First-order orbit queries.
\newblock {\em Theory Comput Syst}, 2020.

\bibitem{basu2005algorithms}
S.~Basu, R.~Pollack, and M-F. Roy.
\newblock {\em Algorithms in real algebraic geometry}, volume 20033.
\newblock Springer, 2005.

\bibitem{chonev2013orbit}
V.~Chonev, J.~Ouaknine, and J.~Worrell.
\newblock The {Orbit Problem} in higher dimensions.
\newblock In {\em Proceedings of the forty-fifth annual ACM symposium on Theory
  of computing}, pages 941--950. ACM, 2013.

\bibitem{chonev2015polyhedron}
V.~Chonev, J.~Ouaknine, and J.~Worrell.
\newblock The {Polyhedron-Hitting Problem}.
\newblock In {\em Proceedings of the Twenty-Sixth Annual ACM-SIAM Symposium on
  Discrete Algorithms}, pages 940--956. SIAM, 2015.

\bibitem{ChonevOW16}
V.~Chonev, J.~Ouaknine, and J.~Worrell.
\newblock On the complexity of the {Orbit Problem}.
\newblock {\em J. {ACM}}, 63(3):23:1--23:18, 2016.

\bibitem{cohen2013course}
H.~Cohen.
\newblock {\em A course in computational algebraic number theory}, volume 138.
\newblock Springer Science \& Business Media, 2013.

\bibitem{TUCS05}
V.~Halava, T.~Harju, M.~Hirvensalo, and J.~Karhum\"aki.
\newblock {Skolem's Problem} -- on the border between decidability and
  undecidability.
\newblock Technical Report 683, Turku Centre for Computer Science, 2005.

\bibitem{harrison1969lectures}
M.~A Harrison.
\newblock Lectures on linear sequential machines.
\newblock Technical report, DTIC Document, 1969.

\bibitem{kannan1980orbit}
R.~Kannan and R.~J. Lipton.
\newblock The {Orbit Problem} is decidable.
\newblock In {\em Proceedings of the twelfth annual ACM symposium on Theory of
  computing}, pages 252--261. ACM, 1980.

\bibitem{kannan1986polynomial}
R.~Kannan and R.~J. Lipton.
\newblock Polynomial-time algorithm for the {Orbit Problem}.
\newblock {\em Journal of the ACM (JACM)}, 33(4):808--821, 1986.

\bibitem{KCBR18}
Zachary Kincaid, John Cyphert, Jason Breck, and Thomas~W. Reps.
\newblock Non-linear reasoning for invariant synthesis.
\newblock {\em Proc. {ACM} Program. Lang.}, 2({POPL}):54:1--54:33, 2018.

\bibitem{LS09}
Martin Leucker and Christian Schallhart.
\newblock A brief account of runtime verification.
\newblock {\em J. Log. Algebr. Program.}, 78(5):293--303, 2009.

\bibitem{MarkeyS03}
Nicolas Markey and Philippe Schnoebelen.
\newblock Model checking a path.
\newblock In {\em {CONCUR} 2003 - Concurrency Theory, 14th International
  Conference, Proceedings}, volume 2761 of {\em Lecture Notes in Computer
  Science}, pages 248--262. Springer, 2003.

\bibitem{mignotte1983some}
M.~Mignotte.
\newblock Some useful bounds.
\newblock In {\em Computer algebra}, pages 259--263. Springer, 1983.

\bibitem{MST84}
M.~Mignotte, T.~Shorey, and R.~Tijdeman.
\newblock The distance between terms of an algebraic recurrence sequence.
\newblock {\em J. f\"ur die reine und angewandte Math.}, 349, 1984.

\bibitem{ouaknine2014ultimate}
J.~Ouaknine and J.~Worrell.
\newblock Ultimate positivity is decidable for simple linear recurrence
  sequences.
\newblock In {\em International Colloquium on Automata, Languages, and
  Programming}, pages 330--341. Springer, 2014.

\bibitem{Rabinovitch}
Alexander Rabinovich and Wolfgang Thomas.
\newblock Decidable theories of the ordering of natural numbers with unary
  predicates.
\newblock In Zolt{\'a}n {\'E}sik, editor, {\em Computer Science Logic}, pages
  562--574, Berlin, Heidelberg, 2006. Springer Berlin Heidelberg.

\bibitem{Renegar88}
J.~Renegar.
\newblock A faster {PSPACE} algorithm for deciding the existential theory of
  the reals.
\newblock In {\em 29th Annual Symposium on Foundations of Computer Science,
  White Plains, New York, USA, 24-26 October 1988}, pages 291--295, 1988.

\bibitem{Semenov}
Alexei Semenov.
\newblock {\em Decidability of monadic theories}, volume 176, pages 162--175.
\newblock Springer Berlin Heidelberg, 04 2006.

\bibitem{tao2008structure}
T.~Tao.
\newblock {\em Structure and randomness: pages from year one of a mathematical
  blog}.
\newblock American Mathematical Soc., 2008.

\bibitem{Ver85}
N.~K. Vereshchagin.
\newblock Occurrence of zero in a linear recursive sequence.
\newblock {\em Mathematical notes of the Academy of Sciences of the USSR},
  38(2):609--615, 1985.

\end{thebibliography}

\end{document}